\documentclass[12pt]{article}
\usepackage{amsmath,amsthm,amsfonts,amssymb}

\usepackage{tikz}
\usepackage{graphicx}
\usepackage{appendix}
\usepackage{subfigure}
\usepackage{xcolor}

\theoremstyle{plain}
\newtheorem{theorem}{Theorem}
\newtheorem{lemma}{Lemma}
\newtheorem{corollary}{Corollary}

\theoremstyle{remark}
\newtheorem{remark}{Remark}
\theoremstyle{definition}
\numberwithin{equation}{section}

\DeclareMathOperator\Ai{{Ai}}
\DeclareMathOperator\Bi{{Bi}}
\newcommand\re{\mathrm{Re}}
\newcommand\im{\mathrm{Im}}
\newcommand{\B}{\mathbf{B}}

\begin{document}
\title{Connection problem of the first Painlev\'{e} transcendents with large initial data}

\author{Wen-Gao Long\thanks{School of Mathematics and Computational Science, Hunan University of Science and Technology, Xiangtan 411201, PR China}\quad
and Yu-Tian Li\thanks{School of Science and Engineering, Chinese University of Hong Kong, Shenzhen, Guangdong, 518172, PR China. Corresponding author.}
}
\date{}

\maketitle

\begin{abstract}

In previous work, Bender and Komijani (2015 \textit{J. Phys. A:
Math. Theor.} 48, 475202) studied the first Painlev\'e (PI) equation
and showed that
the sequence of initial conditions giving rise to separatrix solutions could be asymptotically determined using a $\mathcal{PT}$-symmetric Hamiltonian.
In the present work, we consider the initial value problem of the PI equation
in a more general setting.
We show that the initial conditions $(y(0),y'(0))=(a,b)$ located on a sequence of curves $\Gamma_n$, $n=1,2,\dots$, will give rise to separatrix solutions.
These curves separate the singular and the oscillating solutions of PI.
The limiting form equation $b^2/4 - a^3=f_n \sim A n^{6/5}$ for the curves $\Gamma_{n}$ as $n\to\infty$ is derived, where $A$ is a positive constant.
The discrete set $\{f_n\}$ could be regarded as the nonlinear eigenvalues.
Our analytical asymptotic formula of $\Gamma_n$ matches the numerical results remarkably well,
even for small $n$.
The main tool is the method of uniform asymptotics introduced by Bassom et al.
(1998 \textit{Arch. Rational Mech. Anal.} {143}, 241--271) in the studies of the second Painlev\'e equation.

\end{abstract}


\vspace{5mm}


\section{Introduction}
Painlev\'e equations are a set of six nonlinear second-order ordinary differential equations (ODEs)
possessing the Painlev\'e property, namely,
the movable singularities of solutions must be poles and not branch points or essential singularities.
These ODEs were first studied by Painlev\'e and his colleagues for the classification of
nonlinear equations from a purely mathematical perspective, but
they have found important applications in many areas of mathematical physics,
including statistical physics~\cite{McCoy1992,JMMS1980,EFIK1996,Kan2002,TW2011},
Ising model~\cite{BMW1973,WMTB1976}, and integrable systems~\cite{GRP1991};
see also~\cite{Bender-Komijani-2015} and references therein.

Since there is no exact formula available,
asymptotic analysis is the common practice for studying the behaviors of solutions of Painlev\'e equations, which leads to a natural question -- how to link the behaviors of one solution between different  regimes, such as between asymptotic formulas for $x\to-\infty$ and those for $x\to\infty$. This is known as the connection problems, some particular cases have been studied for the Painlev\'e equations, and others are still awaiting further investigations.

This work focuses on the initial value problem of the first Painlev\'e equation
\begin{align}
&\frac{d^{2}y}{dt^{2}}=6y^2+t, \label{PI equation}
\\
&y(0)=a,\qquad y'(0)=b. \label{initial_condition}
\end{align}
It is known that all PI solutions are irreducible,
that is they cannot be represented by any elementary or classical special function.
The initial value problem was first sutdied by Holmes and Spence \cite{HS-1984},
and they showed that when $a=0$, there exist $\kappa_{1}<0<\kappa_{2}$ such that
the PI solution with $\kappa_{1}<b<\kappa_{2}$ oscillates around the parabola $y=-\sqrt{-t/6}$;
and when $b<\kappa_{1}$ or $b>\kappa_{2}$,
the PI solution possesses infinite number of poles on the negative real axis.
According to Kapaev \cite{AAKapaev-1988} (see also \cite{FAS-2006} and \cite{Long-Li-Liu-Zhao}),
we know that the asymptotic behavior of PI solutions are classified in three types as follows.

\begin{enumerate}
\item [(A)] \emph{oscillating solutions}: a two-parameter family of solutions oscillating about the parabola $y=-\sqrt{-t/6}$;
\item [(B)] \emph{separatrix solutions}: a one-parameter family of solutions satisfying
    \begin{equation}\label{eq-behavior-type-B}
      y(t)=\sqrt{\frac{-t}{6}}+\frac{h}{4\sqrt{\pi}}24^{-\frac{1}{8}}(-t)^{-\frac{1}{8}}\exp\left\{-\frac{4}{5}24^{\frac{1}{4}}(-t)^{\frac{5}{4}}\right\}+\mathcal{O}\left(|t|^{-\frac{5}{2}}\right)
    \end{equation}
    as $t\rightarrow-\infty$, where
    \begin{equation}\label{eq-parameter-h}
     h=s_{1}-s_{4};
    \end{equation}
\item [(C)] \emph{singular solutions}: a two-parameter family of solutions having infinitely many double poles on the negative real axis.
\end{enumerate}
The above three types solutions are characterized by the Stokes multipliers $s_{k}$, $k=0,1,2,3,4$,
\begin{equation}\label{eq-condition-stokes-multipliers}
\begin{aligned}
1+s_{2}s_{3}>0 \quad \text{for type (A)},\\
1+s_{2}s_{3}=0 \quad \text{for type (B)},\\
1+s_{2}s_{3}<0 \quad \text{for type (C)}.
\end{aligned}
\end{equation}
Here, the Stokes multipliers are defined in eq. (\ref{eq-Stokes-matrices}) below,
see also~\cite{Kapaev-Kitaev-1993}.
Type (B), the separatrix solutions, are also known as the \emph{tronqu\'{e}e solutions},
since these solutions are pole-free on two among five sectors in the complex $t$-plane.

A natural question is how to connect the above behaviors with the initial data $(a,b)$,
which is an open problem announced by Clarkson in several occasions \cite{CPA2003, CPA2006, CPA2019}.
Several studies addressed this problem, using asymptotic analysis and numerical simulations \cite{Bender-Komijani-2015,Long-Li-Liu-Zhao,JK-1992,Kitaev-1995,AAKapaev-2004,Fornberg-Weideman}.
In particular, Fornberg and Weidemann \cite{Fornberg-Weideman}
give a phase diagram \cite[Fig. 4.5]{Fornberg-Weideman} of the PI real solutions on the $(a,b)$-plane.
They find that there exists a sequence of curves on the $(a,b)$-plane
giving rise to the tronqu\'{e}e solutions, i.e., Type (B) solutions.
These curves divide the $(a,b)$-plane into separated regions,
and numerical simulations show that the PI solutions alternate between Type (A) and Type (C)
when the initial data $(a,b)$ vary continuously among these regions.
Hence, rigorous proof of this phenomenon
and an analytical study of these separating curves
are awaiting further study and deserve a thorough investigation.

In a previous work, Bender--Komijani~\cite{Bender-Komijani-2015} showed
that for a fixed initial condition, say $y(0)=0$,
there exists a sequence of initial slope $y'(0)=b_n$
giving rise to separatrix solutions;
and similarly, when $y'(0)=0$,
there exists a sequence of initial values $y(0)=a_n$
giving rise to separatrix solutions.
The sequences $\{a_{n}\}$ and $\{b_{n}\}$ are called \emph{nonlinear eigenvalues}.
Motivated by the results of Bender--Komijani~\cite{Bender-Komijani-2015},
Long et al.~\cite{Long-Li-Liu-Zhao} provided a rigorous proof of the
asymptotic formulas for the nonlinear eigenvalues $a_n$ and $b_n$
using the complex WKB method, also known as the method of uniform asymtotics.

What Bender--Komijani~\cite{Bender-Komijani-2015} and Long et al. \cite{Long-Li-Liu-Zhao} did is to study how the PI solutions evolve
when the initial data $(a,b)$ vary along a horizontal line or a vertical line in the phase plane.
What remains and perhaps is more interesting is to see how the solutions evolve
when both $y(0)$ and $y'(0)$ vary at the same time.
To give a characterization of the curves in \cite[Fig. 4.5]{Fornberg-Weideman},
we should consider this more general case instead of fixing one of the initial data.

Noting that the Hamiltonian of PI is $H(t):=\frac{y_{t}^2}{4}+\frac{ty}{2}-y^3$;
hence $H(0)=\frac{b^2}{4}-a^3$.
This quantity suggests us to regard $\frac{b^2}{4}-a^3$ as a whole term --
we only need to assume this term is large, and to derive an asymptotic classification of the PI solutions in terms of this quantity $\frac{b^2}{4}-a^3$.
In this sense, the problem under consideration turns to
be the connection problem of the Hamiltonian $H(t)$ of PI between $t=0$ and $t\to-\infty$
when $H(0)$ is large.

By the complex WKB method (also known as the method of \emph{uniform asymptotics}),
we derive the limiting form equations of the curves in \cite[Fig. 4.5]{Fornberg-Weideman}
as $\frac{b^2}{4}-a^3\to\pm\infty$,
under the assumption that $a^3\left|\frac{b^2}{4}-a^3\right|^{-1}$ is bounded.
In particular, we find that the initial data $(y(0),y'(0))=(a,b)$ located on the sequence of curves
\begin{equation}
\frac{b^2}{4}-a^3= f_n(a,b)\sim A\, n^{6/5},
\qquad A=\left[\frac{5\pi}{2\sqrt{3}\B\left(\frac{1}{2},\frac{1}{3}\right)}\right]^{6/5},
\end{equation}
give rise to the separatrix solution;
cf. Theorem~\ref{thm-1} below.
According to Bender--Komijani~\cite{Bender-Komijani-2015},
the separatrix solutions play the role of eigenfunctions,
and the corresponding initial data can be regarded as the eigenvalues for nonlinear equations.
Thus we may regard the discrete set $\{f_n(a,b)\}$ as \emph{nonlinear eigenvalues} for PI.
One may wonder how the solutions of PI evolve when $\left|\frac{b^2}{4}-a^3\right|\ll a^3$ as $a\to+\infty$.
A particular case of this situation deserves to be mentioned.
Assume $a\to+\infty$ with $\frac{b^2}{4}-a^3=0$.
From \cite[Fig. 4.5]{Fornberg-Weideman},
one might guess that the curve $\Gamma_{0}: \frac{b^2}{4}-a^3=0$ is a limiting form boundary of the fingerprint region when $a\to+\infty$,
that is all the curves in \cite[Fig. 4.5]{Fornberg-Weideman} locate to the left of $\Gamma_{0}$ when $a$ is large enough.
However, numerical evidence tells us that the three types of PI solutions
alternate when $a,b$ increase along $\Gamma_{0}: \frac{b^2}{4}-a^3=0$.
In other words, the curve $\Gamma_{0}$ passes through all the curves plotted in \cite[Fig.~4.5]{Fornberg-Weideman}.
We also find that when $a,b$ are large enough and varying on the curve $\tilde{\Gamma}_{0}: \frac{b^2}{4}-a^3+\frac{b}{4a}=0$,
the above-mentioned alternating phenomenon disappears
and all the PI solutions have infinite number of pole on the negative real axis,
i.e., all solutions belong to Type (C).
The above discussion suggests us to regard $\frac{b^2}{4}-a^3+\frac{b}{4a}$ as a whole term
and to assume it is large when $\left|\frac{b^2}{4}-a^3\right|\ll a^3$.
Basing on this assumption, we succeed to build the limiting-form equation of these curves
in the situation that $\left|\frac{b^2}{4}-a^3\right|\ll a^3$.

The main tool in our analysis is the method of uniform asymptotics
introduced by Bassom et al.~\cite{BCLM} in their studies of connection problem of PII.
The same method has been successively applied to PI in Long et al.~\cite{Long-Li-Liu-Zhao}.
Most part of the analysis of the present work follows from that in~\cite{Long-Li-Liu-Zhao}
with some necessary modifications, see the proof of Lemma~\ref{lem-case-I} below.

In summary, the present paper aims to study the connection problem of PI between $t=0$ and $t=-\infty$ when the initial data is large.
We intend to give an asymptotical classification of the PI solutions in terms of $a$ and $b$
when at least one of them is large.
The major step is to show the existence and to derive a limiting form equation of the curves in \cite[Fig. 4.5]{Fornberg-Weideman}.
The analysis will be divided into two cases:
(i) $a^{3}\leq M_{1}\left|\frac{b^2}{4}-a^3\right|$; and
(ii)$a^{3}\geq M_{2}\left|\frac{b^2}{4}-a^3\right|$,
where $M_{1}$ and $M_{2}$ are two constants with $0<M_{2}<M_{1}$.
In case (i), the result generalizes the main theorems in \cite{Long-Li-Liu-Zhao},
which is similar to the results for the PII equation given in \cite{Long-Zeng}.
In case (ii), we get some new properties of the PI solutions.
In particular, we find how the PI solutions evolve when $a$ increases to infinity
with $\frac{b^2}{4}-a^3=0$.

\section{Monodromy theory of PI}
The analysis in this paper starts from the following Lax pairs for the PI equation
\begin{equation}\label{lax pair-I}
\left\{\begin{aligned}
\frac{\partial\Psi}{\partial\lambda}
&=\left\{(4\lambda^4+t+2y^2)\sigma_{3}
-i(4y\lambda^2+t+2y^2)\sigma_{2}
-\Big(2y_{t}\lambda+\frac{1}{2\lambda}\Big)\sigma_{1}\right\}\Psi
\\
\frac{\partial\Psi}{\partial t}
&=\left\{\Big(\lambda+\frac{y}{\lambda}\Big)\sigma_{3}-\frac{iy}{\lambda}\sigma_{2}\right\}\Psi
\end{aligned}\right.
\end{equation}
where
$$\sigma_{1}=\begin{pmatrix}0&1\\1&0\end{pmatrix},\quad \sigma_{2}=\begin{pmatrix}0&-i\\i&0\end{pmatrix},\quad \sigma_{3}=\begin{pmatrix}1&0\\0&-1\end{pmatrix}$$
are the Pauli matrices and $y_{t}=dy/dt$. The
compatibility condition, $\Psi_{\lambda t}=\Psi_{t\lambda}$,
implies that $y=y(t)$ satisfies the first Painlev\'{e} equation;
see, for example, \cite{Kapaev-Kitaev-1993}.
Under the transformation
\begin{equation}\label{eq-transform-canonical solution}
\Phi(\lambda)
=\lambda^{\frac{1}{4}\sigma_{3}}
\frac{\sigma_{3}+\sigma_{1}}{\sqrt{2}}
\Psi\big(\sqrt{\lambda}\big),
\end{equation}
the first equation of (\ref{lax pair-I}) becomes
\begin{equation}\label{eq-fold-Lax-pair}
\frac{\partial\Phi}{\partial\lambda}
=\begin{pmatrix}y_{t}&2\lambda^{2}+2y\lambda-t+2y^2\\2(\lambda-y)&-y_{t}\end{pmatrix}\Phi.
\end{equation}
It is clear that the only singularity of eq. \eqref{eq-fold-Lax-pair} is
the irregular singular point $\lambda=\infty$,
and the canonical solutions $\Phi_{k}(\lambda)$, $k\in\mathbb{Z}$, satisfy
\begin{equation}\label{eq-canonical-solutions}
\Phi_{k}(\lambda,t)
=\lambda^{\frac{1}{4}\sigma_{3}}\frac{\sigma_{3}+\sigma_{1}}{\sqrt{2}}
\left(I+\frac{\mathcal{H}}{\sqrt{\lambda}}+\mathcal{O}\left(\frac{1}{\lambda}\right)\right)
e^{\big(\frac{4}{5}\lambda^{\frac{5}{2}}+t\lambda^{\frac{1}{2}}\big)\sigma_{3}},
\quad\lambda\in\Omega_{k},
\end{equation}
as $\lambda\to\infty$,
where $\mathcal{H}=-(\frac{1}{2}y_{t}^2-2y^3-ty)\sigma_{3}$,
and
$$\Omega_{k}=\left\{\lambda\in\mathbb{C}~:~\arg \lambda\in \left(-\frac{3\pi}{5}+\frac{2k\pi}{5},\frac{\pi}{5}+\frac{2k\pi}{5}\right)\right\}, \qquad k\in\mathbb{Z}.$$
These canonical solutions are connected by the Stokes matrices,
\begin{equation}\label{eq-Stokes-matrices}
\Phi_{k+1}=\Phi_{k}S_{k},
\end{equation}
where the Stokes matrices are triangular and given as
\begin{equation}
S_{2k-1}=\begin{pmatrix}1&s_{2k-1}\\0&1\end{pmatrix},\quad S_{2k}=\begin{pmatrix}1&0\\s_{2k}&1\end{pmatrix},
\end{equation}
and $s_{k}$ are the Stokes multipliers with the following constraints
\begin{equation}\label{eq-constraints-stokes-multipliers}
s_{k+5}=s_{k}\quad \text{and}\quad s_{k}=i(1+s_{k+2}s_{k+3}),\qquad k\in\mathbb{Z}.
\end{equation}
In addition, regarding $s_{k}$ as functions of $(t,y(t),y'(t))$, they also satisfy
\begin{equation}\label{eq-sk-s-k-relation}
s_{k}\left (t,y(t),y'(t)\right)=-\overline{s_{-k}\left (\bar{t},\overline{y(t)},\overline{y'(t)}\right )}, \qquad k\in\mathbb{Z},
\end{equation}
where $\bar\zeta$ stands for the complex conjugate of a complex number $\zeta$;
see \cite[eq. (13)]{AAKapaev-1988}.
From eq.~\eqref{eq-constraints-stokes-multipliers}, it is readily seen that, in general, two of the Stokes multipliers determine all others.
For the derivation  of eqs. (\ref{eq-canonical-solutions}), (\ref{eq-Stokes-matrices}) and (\ref{eq-constraints-stokes-multipliers}), and more details about the Lax pairs,
the reader is referred to \cite{FAS-2006}.

The main technique we use to calculate the Stokes multipliers is the complex WKB method,
also known as the method of \emph{uniform asymptotics} introduced by Bassom et al.~\cite{BCLM}.
The derivation consists of two steps. In the first step, we transform the first equation of the Lax pair \eqref{lax pair-I} into a second-order Shr\"{o}dinger equation. Denoting
$$\Phi(\lambda)=\left(\begin{matrix}\phi_{1}\\\phi_{2}\end{matrix}\right),$$
and defining $Y(\lambda;t)=(2(\lambda-y))^{-\frac{1}{2}}\phi_{2}$, we have
\begin{equation}\label{Schrodinger-equation-t-general}
\frac{d^{2}Y}{d\lambda^{2}}=\left[y_{t}^{2}+4\lambda^{3}+2\lambda t-2y t-4y^{3}-\frac{y_{t}}{\lambda-y}+\frac{3}{4}\frac{1}{(\lambda-y)^2}\right]Y.
\end{equation}
When $t=0$, equation (\ref{Schrodinger-equation-t-general}) is simplified as
\begin{equation}\label{Schrodinger-equation}
\frac{d^{2}Y}{d\lambda^{2}}
=\left[4\lambda^{3}+4\left(\frac{b^{2}}{4}-a^{3}\right)
-\frac{b}{\lambda-a}+\frac{3}{4}\frac{1}{(\lambda-a)^2}\right] Y.
\end{equation}
One may regard (\ref{Schrodinger-equation}) as either a scalar or a $1\times 2$ vector equation.
In the whole paper, we assume that at least one of $a$ and $b$ is large.

In case (i), $a^3\leq M_{1}\left|\frac{b^2}{4}-a^3\right|$,
since we assume at least one of $a$ and $b$ is large,
we conclude that $\frac{b^2}{4}-a^3$ is large positive or negative.
To apply the method of uniform asymptotics,
we introduce a large parameter
$\xi\to+\infty$ with
\begin{equation}\label{eq:xi}
\frac{b^2}{4}-a^3=\pm\xi^{\frac{6}{5}}
\end{equation}
and set
\begin{equation}\label{eq:transform1}
\lambda=\eta\xi^{\frac{2}{5}},\qquad
a=A(\xi)\xi^{\frac{2}{5}}, \quad
b=B(\xi)\xi^{\frac{3}{5}}.
\end{equation}
Then,
\begin{equation}\label{eq-two-cases-situation-I}
\begin{cases}
\frac{B(\xi)^{2}}{4}-A(\xi)^{3}&=1,\\
\frac{B(\xi)^{2}}{4}-A(\xi)^{3}&=-1.
\end{cases}
\end{equation}
and $A(\xi),\,B(\xi)=\mathcal{O}(1)$ as $\xi\to+\infty$. Moreover, equation \eqref{Schrodinger-equation} becomes
\begin{equation}\label{Schrodinger-equation-scaled}
\begin{aligned}
\frac{d^{2}Y}{d\eta^{2}}
=&\xi^2\left[4\eta^{3}+4\left(\frac{B(\xi)^{2}}{4}-A(\xi)^{3}\right)
-\frac{B(\xi)}{\xi(\eta-A(\xi))}+\frac{3}{4}\frac{1}{\xi^2(\eta-A(\xi))^2}\right]Y
\\
:=&\xi^2\left[4\eta^{3}+4\left(\frac{B(\xi)^{2}}{4}-A(\xi)^{3}\right)
-\varphi(\eta,\xi)\right]Y,
\end{aligned}
\end{equation}
where $\varphi(\eta,\xi)=\mathcal{O}(\xi^{-1})$ as $\xi\to+\infty$ uniformly
for all $\eta$ away from $\eta=A(\xi)$ and all bounded $A(\xi),B(\xi)$.

In case (ii), when $a^{3}\gg \left|\frac{b^2}{4}-a^3\right|$,
we have $b^2\sim 4a^3$;
thus $\frac{b}{a}\sim 2a^{\frac{1}{2}}\gg \xi^{\frac{1}{5}}$ as $\xi\to+\infty$,
and maybe dominate $\left|\frac{b^2}{4}-a^3\right|$.
Hence the large parameter used in case (i) is not appropriate anymore,
and instead, we introduce another large parameter $\mu\to+\infty$ with
\begin{equation}\label{eq:mu}
\frac{b^2}{4}-a^3+\frac{b}{4a}=\pm\mu^{\frac{6}{5}}
\end{equation}
and set
\begin{equation}\label{eq:transform2}
\lambda=\tilde\eta\mu^{\frac{2}{5}},\qquad
a=\tilde{A}(\mu)\mu^{\frac{2}{5}},\quad
b=\tilde{B}(\mu)\mu^{\frac{3}{5}}.
\end{equation}
Then
\begin{equation}\label{eq-two-cases-situation-II}
\begin{cases}
\frac{\tilde{B}(\mu)^{2}}{4}-\tilde{A}(\mu)^{3}+\frac{\tilde{B}(\mu)}{4\tilde{A}(\mu)}&=1,\\
\frac{\tilde{B}(\mu)^{2}}{4}-\tilde{A}(\mu)^{3}+\frac{\tilde{B}(\mu)}{4\tilde{A}(\mu)}&=-1,
\end{cases}
\end{equation}
and $|\tilde{A}(\mu)|^{-1},\,|\tilde{B}(\mu)|^{-1}=\mathcal{O}(1)$ as $\mu\to+\infty$.
In terms of the new variable $\tilde{\eta}$, equation \eqref{Schrodinger-equation} becomes
\begin{equation}\label{Schrodinger-equation-scaled}
\begin{aligned}
\frac{d^{2}Y}{d\tilde\eta^{2}}
=&\mu^2\left[4\tilde\eta^{3}+4\left(\frac{\tilde{B}(\mu)^{2}}{4}-\tilde{A}(\mu)^{3}+\frac{\tilde{B}(\mu)}{4\mu\tilde{A}(\mu)}\right)\right.
\\
&\hspace{7mm}-\left.\frac{\tilde{B}(\mu)\tilde\eta/\tilde{A}(\mu)}{\mu(\tilde\eta-\tilde{A}(\mu))}+\frac{3/4}{\mu^2(\tilde\eta-\tilde{A}(\mu))^2}\right]Y
\\
&=:\mu^2\left[4\tilde\eta^{3}+4\left(\frac{\tilde{B}(\mu)^{2}}{4}-\tilde{A}(\mu)^{3}
+\frac{\tilde{B}(\mu)}{4\mu\tilde{A}(\mu)}\right)-\tilde{\varphi}(\tilde\eta,\mu)\right] Y,
\end{aligned}
\end{equation}
where $\tilde{\varphi}(\tilde\eta,\mu)=\mathcal{O}(\mu^{-1})$ as $\mu\to+\infty$ uniformly for all $\tilde\eta$ away from $\tilde\eta=\tilde{A}(\mu)$ and all $|\tilde{A}(\mu)|^{-1},|\tilde{B}(\mu)|^{-1}=\mathcal{O}(1)$.

In the follow-up analysis, our task is to derive the uniform asymptotic behaviours of the solutions of this equation in the above two situations. In each situation, we obtain the asymptotic behaviour of the Stokes multipliers as $\xi\to+\infty$ or $\mu\to+\infty$. Substituting the corresponding results into \eqref{eq-behavior-type-B}, we get the asymptotic classification of the PI solutions in terms of $a$ and $b$.

The rest of the paper is organized as follows.
In section 2, we state our main results
including the leading asymptotic behavior of the Stokes multipliers
and the asymptotic classification of the PI solutions in terms of $a$ and $b$.
The proof of some lemmas are given in section 3,
where the method of uniform asymptotics is carried out.
A concluding remark and a discussion are given in the last section.

\section{Main results}
Recall that the three types of solutions of PI
are classified by the Stokes multipliers in eq.~\eqref{eq-condition-stokes-multipliers},
the main task for us is to derive the relationship
between the Stokes multipliers and the initial data $a$ and $b$.
It is difficult to find exact formulas for the Stokes multipliers in terms of $a$ and $b$,
however, motivated by the ideas in \cite{Long-Li-Liu-Zhao} and \cite{Sibuya-1967},
it seems possible to get the leading asymptotic behavior of the Stokes multipliers.

\subsection{Leading asymptotic behavior of the Stokes multipliers}

When the initial data satisfy $a^3\leq M_{1}\left|\frac{b^2}{4}-a^3\right|$,
we find that $A(\xi),B(\xi)=\mathcal{O}(1)$ as $\xi\to+\infty$.
According to \eqref{eq-two-cases-situation-I},
we shall divide the discussion into two cases,
and the corresponding results are stated in the following two Lemmas respectively.
\begin{lemma}\label{lem-case-I}
When $\frac{B(\xi)^{2}}{4}-A(\xi)^{3}=1$, {\it i.e.} $\frac{b^2}{4}-a^3=\xi^{\frac{6}{5}}\to+\infty$, we have
\begin{equation}\label{asym-Stokes-multipliers-case-I}
\begin{aligned}
s_{1}=&i\exp\left\{2\xi E_{0}-2\, E_{1}+o(1)\right\},\\
s_{-1}=-\overline{s_{1}}=&i\exp\left\{2\xi F_{0}-2\, F_{1}+o(1)\right\}
\end{aligned}
\end{equation}
uniformly for all $A(\xi),B(\xi)=\mathcal{O}(1)$, where
\begin{equation}\label{eq-constants-lemma-1}
\begin{aligned}
E_{0}&=\overline{F_{0}}=\frac{3}{5}\B\left(\frac{1}{2},\frac{1}{3}\right)-\frac{\sqrt{3}i}{5}\B\left(\frac{1}{2},\frac{1}{3}\right),\\
E_{1}&=\overline{F_{1}}=\frac{1}{4}\int_{e^{\frac{\pi i}{3}}}^{\infty e^{\frac{\pi i}{3}}}\frac{B(\xi)}{\left(s^3+1\right)^{\frac{1}{2}}(s-A(\xi))}ds.
\end{aligned}
\end{equation}
\end{lemma}
\begin{lemma}\label{lem-case-II}
When $\frac{B(\xi)^{2}}{4}-A(\xi)^{3}=-1$, {\it i.e.} $\frac{b^2}{4}-a^3=-\xi^{\frac{6}{5}}\to-\infty$, we have
\begin{equation}\label{asym-Stokes-multipliers-case-II}
\begin{aligned}
s_{2}=&i \exp\left\{-2\xi G_{0}+2G_{1}+o(1)\right\},\\
s_{3}=s_{-2}=&i \exp\left\{-2\xi H_{0}+2H_{1}+o(1)\right\}
\end{aligned}
\end{equation}
uniformly for all $A(\xi),B(\xi)=\mathcal{O}(1)$,
where
\begin{equation}
\begin{aligned}
G_{0}=\overline{H_{0}}
=&\frac{2}{5}e^{\frac{2\pi i}{3}}\B\left(\frac{1}{2},\frac{1}{3}\right),
\\
G_{1}=\overline{H_{1}}
=&\frac{1}{4}\int_{e^{\frac{2\pi i}{3}}}^{\infty e^{\frac{2\pi i}{3}}}
\frac{B(\xi)}{\left(s^3-1\right)^{\frac{1}{2}}(s-A(\xi))}ds.
\end{aligned}
\end{equation}
\end{lemma}

When the initial data satisfy $a^{3}\geq M_{2}\left|\frac{b^2}{4}-a^3\right|$,
we also divide the discussion into two cases according to \eqref{eq-two-cases-situation-II},
and state the corresponding leading asymptotic behavior of the Stokes multipliers in the following Lemmas.

\begin{lemma}\label{lem-case-I-situation-II}
When $\frac{\tilde{B}(\mu)^{2}}{4}-\tilde{A}(\mu)^{3}+\frac{\tilde{B}(\mu)}{4\mu\tilde{A}(\mu)}=1$, {\it i.e.} $\frac{b^2}{4}-a^3+\frac{b}{4a}=\mu^{\frac{6}{5}}\to+\infty$, we have
\begin{equation}\label{asym-Stokes-multipliers-case-I-situation-II}
\begin{aligned}
s_{1}=&i\exp\left\{2\mu E_{0}-2\, \tilde{E}_{1}+o(1)\right\},\\
s_{-1}=-\overline{s_{1}}=&i\exp\left\{2\mu F_{0}-2\, \tilde{F}_{1}+o(1)\right\}
\end{aligned}
\end{equation}
uniformly for all $|\tilde{A}(\mu)|^{-1},|\tilde{B}(\mu)|^{-1}=\mathcal{O}(1)$, where $E_{0}, F_{0}$ are given in \eqref{eq-constants-lemma-1} and
\begin{equation}\label{eq-constants-lemma-3}
\tilde{E}_{1}
=\overline{\tilde{F}_{1}}
=\frac{1}{4}\int_{e^{\frac{\pi i}{3}}}^{\infty e^{\frac{\pi i}{3}}}
\frac{s\tilde{B}(\mu)}{\left(s^3+1\right)^{\frac{1}{2}}\tilde{A}(\mu)(s-\tilde{A}(\mu))}ds.
\end{equation}
\end{lemma}

\begin{lemma}\label{lem-case-II-situation-II}
When $\frac{\tilde{B}(\mu)^{2}}{4}-\tilde{A}(\mu)^{3}+\frac{\tilde{B}(\mu)}{4\mu\tilde{A}(\mu)}=-1$, {\it i.e.} $\frac{b^2}{4}-a^3+\frac{b}{4a}=-\mu^{\frac{6}{5}}\to-\infty$, we have
\begin{equation}\label{asym-Stokes-multipliers-case-II-situation-II}
\begin{aligned}
s_{2}=&i \exp\left\{-2\mu G_{0}+2\tilde{G}_{1}+o(1)\right\},\\
s_{3}=-\overline{s_{2}}=&i \exp\left\{-2\mu H_{0}+2\tilde{H}_{1}+o(1)\right\}
\end{aligned}
\end{equation}
uniformly for all $|\tilde{A}(\mu)|^{-1},|\tilde{B}(\mu)|^{-1}=\mathcal{O}(1)$,
where $G_{0}, H_{0}$ are given in Lemma~\ref{lem-case-II} and
\begin{equation}
\tilde{G}_{1}=\overline{\tilde{H}_{1}}
=\frac{1}{4}\int_{e^{\frac{2\pi i}{3}}}^{\infty e^{\frac{2\pi i}{3}}}
\frac{\tilde{B}(\mu)s}{\left(s^3-1\right)^{\frac{1}{2}}\tilde{A}(\mu)(s-\tilde{A}(\mu))}ds.
\end{equation}
\end{lemma}

\begin{remark}\label{remark-1}
(i) When $a$ or $b\to\infty$ with $M_{2}\leq a^3\left|\frac{b^2}{4}-a^3\right|^{-1}\leq M_{1}$,
we have $\left|\frac{b^2}{4}-a^3\right|=\xi^{\frac{6}{5}}$ and $\frac{b}{4a}=\mathcal{O}(\xi^{\frac{1}{5}})$.
Hence $\frac{b}{4a}\ll\frac{b^2}{4}-a^3$,
which implies that $\xi\sim\mu$ and
\begin{equation}
\begin{aligned}
\mu E_{0}-\tilde{E}_{1}=&\xi E_{0}-E_{1}+o(1),\\
\mu F_{0}-\tilde{F}_{1}=&\xi F_{0}-F_{1}+o(1)
\end{aligned}
\end{equation}
as $\xi\to+\infty$.

(ii) When $\frac{b^2}{4}-a^3\to-\infty$ with $M_{2}\leq a^3\left|\frac{b^2}{4}-a^3\right|^{-1}\leq M_{1}$, we also have $\xi\sim\mu$ and
\begin{equation}
\begin{aligned}
\mu G_{0}-\tilde{G}_{1}=\xi G_{0}-G_{1}+o(1),\\
\mu H_{0}-\tilde{H}_{1}=\xi H_{0}-H_{1}+o(1).
\end{aligned}
\end{equation}
This means that in the overlapping regions $M_{2}\leq a^3\left|\frac{b^2}{4}-a^3\right|^{-1}\leq M_{1}$ of $(a,b)$,
Lemmas \ref{lem-case-I} and \ref{lem-case-II} are consistent with Lemmas \ref{lem-case-I-situation-II} and \ref{lem-case-II-situation-II}, respectively.
\end{remark}

\subsection{Classification of the PI solutions}
A combination of eqs. \eqref{asym-Stokes-multipliers-case-I}, \eqref{asym-Stokes-multipliers-case-I-situation-II} and \eqref{eq-condition-stokes-multipliers}
leads to the following result.
\begin{theorem}\label{thm-1}
Suppose $y(t; a, b)$ is a solution of PI equation with $y(0)=a,y'(0)=b$
and $M_{1},M_{2}$ are two arbitrary constants with $0<M_{2}<M_{1}$.
Then there exists a sequence of curves
\begin{equation}\label{eq-Gamma-n}
\Gamma_{n}:\begin{cases}\frac{b^2}{4}-a^3=f_{n}(a,b), & a^3\leq M_{1}\left|\frac{b^2}{4}-a^3\right|,\\[2mm]
\frac{b^2}{4}-a^3+\frac{b}{4a}=g_{n}(a,b),& a^{3}\geq M_{2}\left|\frac{b^2}{4}-a^3\right|
\end{cases}
\end{equation}
with
\begin{equation}
f_{n}(a,b)
=\left[\frac{5}{2\sqrt{3}\B\left(\frac{1}{2},\frac{1}{3}\right)}
\left(n\pi-\frac{\pi}{2}-\im(E_{1}-F_{1})+o(1)\right)\right]^{\frac{6}{5}}
\end{equation}
and
\begin{equation}
g_{n}(a,b)
=\left[\frac{5}{2\sqrt{3}\B\left(\frac{1}{2},\frac{1}{3}\right)}
\left(n\pi-\frac{\pi}{2}-\im(\tilde{E}_{1}-\tilde{F}_{1})+o(1)\right)\right]^{\frac{6}{5}},
\end{equation}
as $n\to\infty$, such that the corresponding solutions $y(t; a,b)$ with $(a,b)\in\Gamma_{n}$ belong to Type (B) with
\begin{equation}\label{eq-asym-hn}
h:=h_{n}=\begin{cases}
(-1)^{n-1}2\exp\left\{f_{n}(a,b)P_{0}+o(1)\right\},&\quad a^3\leq M_{1}\left|\frac{b^2}{4}-a^3\right|, \\[2mm]
(-1)^{n-1}2\exp\left\{g_{n}(a,b)P_{0}+o(1)\right\},&\quad a^{3}\geq M_{2}\left|\frac{b^2}{4}-a^3\right|
\end{cases}
\end{equation}
as $n\to+\infty$, where $P_{0}=\frac{6}{5}\B\left(\frac{1}{2},\frac{1}{3}\right)$ and $E_{1},F_{1},\tilde{E}_{1},\tilde{F}_{1}$ are given in \eqref{eq-constants-lemma-1} and \eqref{eq-constants-lemma-3}.
Moreover, we have
\begin{itemize}
\item [$\mathrm{(i)}$] If $(a,b)$ lies in the region between $\Gamma_{2m}$ and $\Gamma_{2m+1}$,
the solution $y(t; a,b)$ belongs to Type (A),
\item [$\mathrm{(ii)}$] If $(a,b)$ lies in the region between $\Gamma_{2m-1}$ and $\Gamma_{2m}$,
the solution $y(t; a,b)$ belongs to Type (C).
\end{itemize}
\end{theorem}

\begin{figure}[!ht]
\label{fig-initial-real-tronquee}
\centering\includegraphics[width=0.7\textwidth]{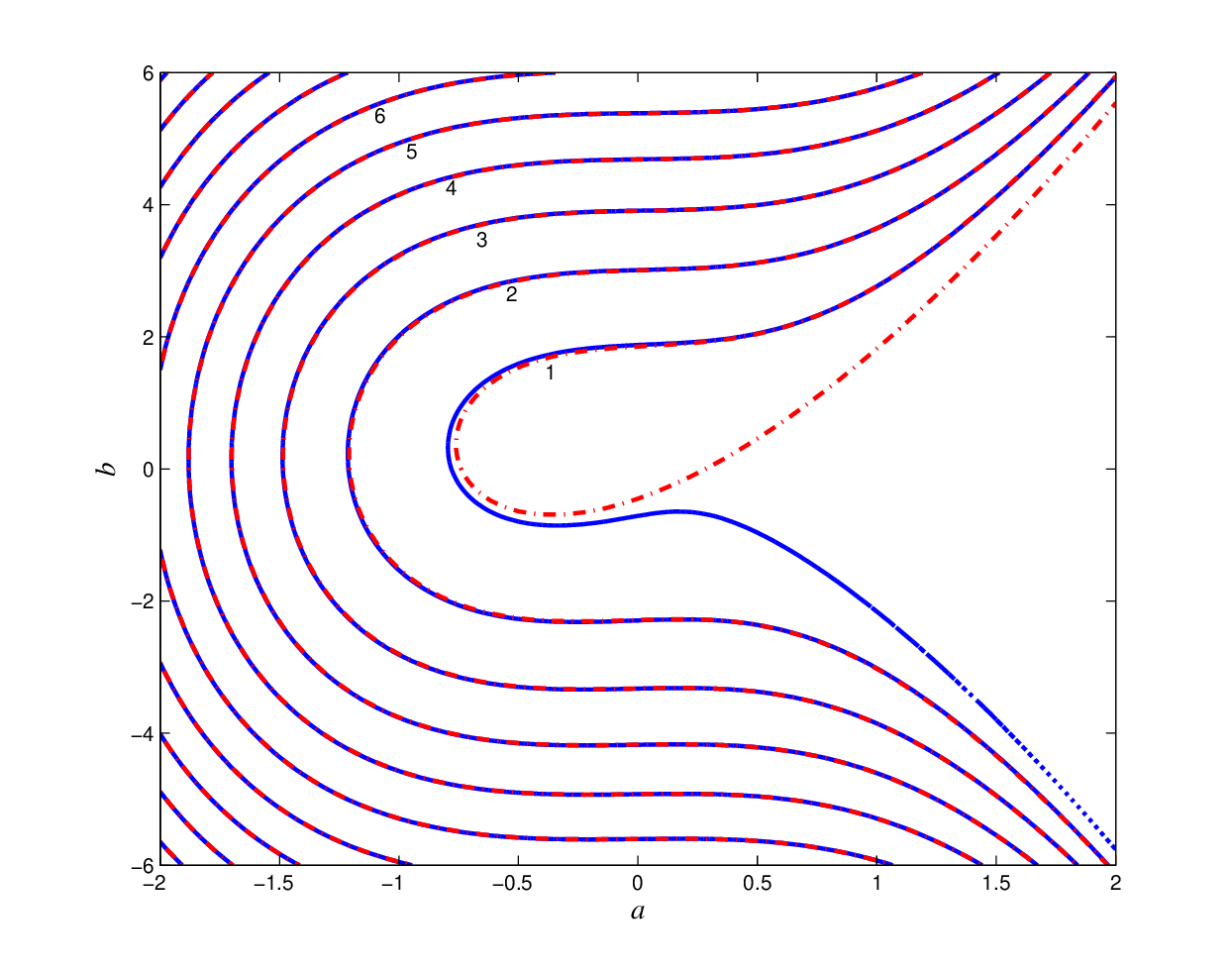}
  \caption{The curves $\Gamma_{n}$ of $(a,b)$ which give rise to the real tronqu\'{e}e solutions of PI. The blue lines are the large-$n$ asymptotics of $\Gamma_{n}$ given in \eqref{eq-Gamma-n},
  and the red dash-dotted curves are the exact values of $\Gamma_{n}$ obtained by numerical simulations.}
\end{figure}

\begin{proof}
We shall assume $a^3\leq M_{1}\left|\frac{b^2}{4}-a^3\right|$ here,
and the case $a^{3}\geq M_{2}\left|\frac{b^2}{4}-a^3\right|$ can be proved by a similar manner.

Noting that $s_{k}=i(1+s_{k+2}s_{k+3})$ from (\ref{eq-constraints-stokes-multipliers}),
we have
\begin{equation}\label{eq-s0-by-s1-s-1}
s_{0}=\frac{s_{1}+s_{-1}-i}{-s_{1}s_{-1}}.
\end{equation}
Substituting the leading asymptotic behaviour of $s_{1}$ and $s_{-1}$ as $\xi\to+\infty$ given in Lemma \ref{lem-case-I} into \eqref{eq-s0-by-s1-s-1},
we get
\begin{equation}\label{eq-asymp-s0}
s_{0}=ie^{\re(E_{0}+F_{0})\xi+\re(E_{1}+F_{1})+o(1)}\left[\cos\left(-\im(E_{0}-F_{0})\xi+\im(E_{1}-F_{1})\right)+o(1)\right]
\end{equation}
as $\xi\to+\infty$.
Regarding $s_0=s_0(\xi)$ as a function of $\xi$,
It is clear that $s_{0}$ has a sequence of zeros.
Let$\{\xi_{n}\}$ denote the zero of $s_{0}(\xi)$ that is near the $n$-th zero of the $\cos$ function in \eqref{eq-asymp-s0} on the positive real axis.
Then,
\begin{equation}
\xi_{n}=\frac{n\pi-\frac{\pi}{2}-\im(E_{1}-F_{1})+o(1)}{-\im(E_{0}-F_{0})},\qquad n\to+\infty.
\end{equation}
This proves the existence of $\Gamma_{n}$.
Substituting \eqref{asym-Stokes-multipliers-case-I} into \eqref{eq-parameter-h}
gives the asymptotic of $h_{n}$ in \eqref{eq-asym-hn}.
Moreover, we also have
\begin{equation}
\im s_{0}\begin{cases}>0,&\quad \xi\in(\xi_{2m},\xi_{2m+1}),\\
<0, & \quad \xi\in(\xi_{2m-1},\xi_{2m}),
\end{cases}
\end{equation}
which, together with \eqref{eq-condition-stokes-multipliers}, implies (i) and (ii);
thus completes the proof.
\end{proof}

\begin{remark}
In Theorem \ref{thm-1}, we only show the existence of $\Gamma_{n}$
and derive the limiting form equation of $\Gamma_{n}$ as $n\to+\infty$.
This is because we have get the leading asymptotic behavior of $s_{0}(\xi)$ as $\xi\to+\infty$. Although the limiting form equations of $\Gamma_{n}$ as $n\to+\infty$ is seperated into two parts for $a^3\leq M_{1}\left|\frac{b^2}{4}-a^3\right|$ and $a^{3}\geq M_{2}\left|\frac{b^2}{4}-a^3\right|$ respectively, it can be seen from Remark \ref{remark-1} that they are consistent in the overlapping region $M_{2}\leq a^3\left|\frac{b^2}{4}-a^3\right|^{-1}\leq M_{1}$.  It should be noted that we can not say $\xi_{n}$ is exactly the $n$-th zero of $s_{0}(\xi)$
since we only obtain the asymptotic behavior of $s_{0}(\xi)$ as $\xi\to+\infty$ in \eqref{eq-asymp-s0}.
Nonetheless, the numerical evidence shows $\xi_n$ is indeed the $n$-th zero of $s_{0}(\xi)$.
This is a coincidence.
\end{remark}

According to Lemmas \ref{lem-case-II} and \ref{lem-case-II-situation-II}, we obtain the following theorem.
\begin{theorem}\label{thm-2}
Let $y(t; a, b)$ be a solution of PI equation with $y(0)=a,y'(0)=b$, then there exist
two constants $M_{3}, M_{4}>0$ such that $y(t; a, b)$ belong to Type (C) provided that $a^3-\frac{b^2}{4}-\frac{b}{4a}>M_{3}$ and $a>M_{4}$.
\end{theorem}

\begin{proof}
According to \eqref{eq-condition-stokes-multipliers}, we only need to show that $\im s_{0}<0$
as $a\to+\infty$ and $\frac{b^2}{4}-a^3+\frac{b}{4a}\to-\infty$.
Since $s_{0}=i(1+s_{2}s_{3})=i(1-|s_{2}|)$ from \eqref{eq-constraints-stokes-multipliers},
it is equivalent to show that $|s_{2}|>1$ when $a$ and $a^3-\frac{b^2}{4}-\frac{b}{4a}$ are both large positive.

When $a\to+\infty$ with $a^3\leq M_{1}\left|\frac{b^2}{4}-a^3\right|$,
we have $\left|\frac{b^2}{4}-a^3\right|\to+\infty$ and $\frac{b}{4a}\ll\left|\frac{b^2}{4}-a^3\right|$,
which, together with $a^3-\frac{b^2}{4}-\frac{b}{4a}\to+\infty$, implies that $\frac{b^2}{4}-a^3\to-\infty$.
According to Lemma \ref{lem-case-II}, we conclude that $|s_{2}|>1$.

When $a\to+\infty$ with $a^{3}\geq M_{2}\left|\frac{b^2}{4}-a^3\right|$,
it immediately follows from Lemma \ref{lem-case-II-situation-II} that $|s_{2}|>1$ when $a^3-\frac{b^2}{4}-\frac{b}{4a}$ is large enough.

This completes the proof of Theorem~\ref{thm-2}.
\end{proof}

The above Theorems \ref{thm-1} and \ref{thm-2} are the major results of the current work.
They contain several important results in the previous work as special cases:
(i) when one of $a$ and $b$ is fixed and the other increases,
we can deduce the corresponding results in \cite[Theorem 1]{Long-Li-Liu-Zhao} and \cite[Theorem 2]{Long-Li-Liu-Zhao} from Theorem \ref{thm-1};
and (ii) \cite[Theorem 3]{Long-Li-Liu-Zhao} is a direct consequence of Theorem \ref{thm-2} above.

Theorems \ref{thm-1} and \ref{thm-2} could also provide new results and phenomena.
For instance, as particular cases of Theorems \ref{thm-1} and \ref{thm-2}, we find how the PI solutions evolve when $(a,b)$ travels along the curve $\frac{b^2}{4}-a^3=0$. The two cases $b\to+\infty$ and $b\to-\infty$ have a significant difference.
\begin{corollary}\label{cor-1}
There exists a positive increasing sequence $0<c_{1}<c_{2}<\cdots$ with
\begin{equation}\label{eq-cn-asym}
c_{n}=4\left[\frac{5\sqrt{3}}{6\B\left(\frac{1}{2},\frac{1}{3}\right)}
\left(n\pi+\frac{\pi}{2}+o(1)\right)\right]^{\frac{12}{5}},\quad n\to+\infty
\end{equation}
such that $y(t; c_{n},2(c_{n})^{3/2})$ belongs to Type (B).
In addition, the solution $y(t; a, 2a^{3/2})$ is of Type (A) when $a\in(c_{2m},c_{2m+1})$,
and of Type (C) when $a\in(c_{2m-1},c_{2m})$.
\end{corollary}

\begin{figure}[!ht]
\label{fig-initial-real-tronquee}
\centering\includegraphics[width=1\textwidth]{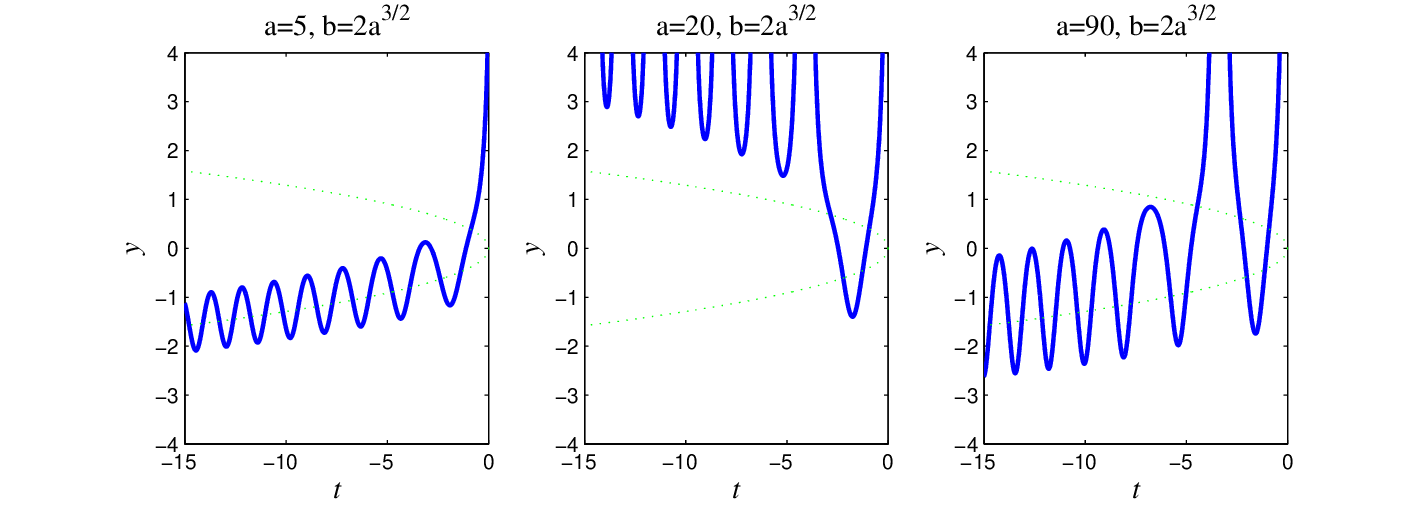}
  \caption{The PI solutions with $a=5,20,90$ and $b=2a^{\frac{3}{2}}$.}
\end{figure}

\begin{proof}
When $b\to+\infty$ with $\frac{b^2}{4}-a^3=0$,
we have $a\to+\infty$ and $a^{-3}\left|\frac{b^2}{4}-a^3\right|\leq M_2$,
which implies that $b\sim 2a^{\frac{3}{2}}$.
According to Theorem \ref{thm-1},
there exists a sequence $\{c_{n}\}$ with
\begin{equation}\label{eq-cn-asymptotic}
c_{n}=4\left[\frac{5\sqrt{3}}{6\B\left(\frac{1}{2},\frac{1}{3}\right)}\left(n\pi-\frac{\pi}{2}-\im(\tilde{E}_{1}-\tilde{F}_{1})+o(1)\right)\right]^{\frac{12}{5}},\quad n\to+\infty
\end{equation}
such that the solution $y(t; c_{n}, 2(c_{n})^{3/2})$ belongs to Type (B).

Since we are assuming $\frac{b^2}{4}-a^3=0$,
using $\frac{b^2}{4}-a^3+\frac{b}{4a}=\mu^{\frac{6}{5}}$ from \eqref{eq:mu},
we find that $\frac{b}{4a}=\frac{a^{\frac{1}{2}}}{2}=\mu^{\frac{6}{5}}$.
It then follows from \eqref{eq:transform2} that
$\tilde{A}(\mu)\sim 4\mu^{2}$ and $\tilde{B}(\mu)\sim 16\mu^{3}$ as $\mu\to+\infty$.
A tedious but straightforward calculation from \eqref{eq-constants-lemma-3} yields
\begin{equation}
\begin{split}
\tilde{E}_{1}-\tilde{F}_{1}&=\frac{1}{4}\int_{e^{\frac{\pi i}{3}}}^{\infty e^{\frac{\pi i}{3}}}\frac{s\tilde{B}(\mu)/\tilde{A}(\mu)}{(s^3+1)^{\frac{1}{2}}(s-\tilde{A}(\mu))}ds-\frac{1}{4}\int_{e^{-\frac{\pi i}{3}}}^{\infty e^{-\frac{\pi i}{3}}}\frac{s\tilde{B}(\mu)/\tilde{A}(\mu)}{(s^3+1)^{\frac{1}{2}}(s-\tilde{A}(\mu))}ds\\
&=-\frac{2\pi i \tilde{B}(\mu)}{4(\tilde{A}(\mu)^3+1)^{\frac{1}{2}}}-\frac{1}{4}\int_{e^{-\frac{\pi i}{3}}}^{e^{\frac{\pi i}{3}}}\frac{s\tilde{B}(\mu)/\tilde{A}(\mu)}{(s^3+1)^{\frac{1}{2}}(s-\tilde{A}(\mu))}ds\\
&=-\pi i+\mathcal{O}(\mu^{-1})
\end{split}
\end{equation}
as $\mu\to+\infty$. Substituting the last equation into \eqref{eq-cn-asymptotic},
we immediately get \eqref{eq-cn-asym}.

The other parts in Corollary \ref{cor-1} are direct consequences of Theorem~\ref{thm-1}.
\end{proof}

\begin{corollary}\label{cor-2}
There exists a constant $M_{5}>0$ such that, whenever $a>M_{5}$,
the corresponding $y(t; a, -2a^{\frac{3}{2}})$ is a Type (C) solution.
\end{corollary}
\begin{proof}
According to Theorem \ref{thm-2}, there exist two positive constants $M_{2}$ and $M_{3}$ such that $y(t; a, b)$ belongs to Type (C) provided that $a^3-\frac{b^2}{4}-\frac{b}{4a}>M_{2}$ and $a>M_{3}$. When $b=-2a^{\frac{3}{2}}$, the first inequality is equivalent to $\frac{1}{2}a^{\frac{1}{2}}>M_{2}$. Hence, Corollary \ref{cor-2} follows if we choose some $M_{4}>\max\{4M_{2}^2,M_{3}\}$.
\end{proof}

\section{Uniform asymptotics and proof of the Lemmas}
This section focus on the proof of Lemmas \ref{lem-case-I} and \ref{lem-case-II}.
While Lemmas \ref{lem-case-I-situation-II} and \ref{lem-case-II-situation-II}
can be shown in a similar manner except for a few minor modifications.
The method of uniform asymptotics is carried out in details
when $\frac{b^2}{4}-a^3\to\infty$ with $a^3\leq M_{1}\left|\frac{b^2}{4}-a^3\right|$,
and the analysis is divided into two cases.

\subsection{Case I: $\frac{B(\xi)^2}{4}-A(\xi)^3=1$}
In this case, equation \eqref{Schrodinger-equation-scaled} becomes
\begin{equation}\label{Schrodinger-equation-scaled-I}
\begin{aligned}
\frac{d^{2}Y}{d\eta^{2}}
=&\xi^2\left[4(\eta^{3}+1)-\frac{B(\xi)}{\xi(\eta-A(\xi))}+\frac{3}{4}\frac{1}{\xi^2(\eta-A(\xi))^2}\right]Y\\
=&\xi^2 Q(\eta,\xi)Y,
\end{aligned}
\end{equation}
which is a Shr\"{o}dinger equation with three simple turning points,
denoted by $\eta_{j},j=1,2,3$, near $\eta=e^{\pi i/3},\,e^{-\pi i/3}$ and $-1$ respectively;
see Figure \ref{Figure-stokes}.
According to \cite{EGS-2008}, the limiting state of the Stokes geometry of the quadratic form $Q(\eta,\xi)d\eta^2$ as $\xi\rightarrow+\infty$ is described in Figure \ref{Figure-stokes},
where the Stokes curves are defined by $\re(\int_{\eta_{j}}^{\eta}\sqrt{Q(s,\xi)}ds)=0$, $j=1,2,3$. Note that the potential $Q(\eta,\xi)$ also has a pole, denoted by $P$, at $\eta=A(\xi)$,
and $P$ can appear to the left or the right of the curve $\ell_{12}$.
When $P$ is to the right of $\ell_{12}$,
it seems difficult to obtain the uniform asymptotic behavior of $Y$ in the whole sector $\arg{\lambda}\in\left[-\frac{2\pi}{5},\frac{2\pi}{5}\right]$.
Hence, we cannot derive the Stokes multiplier $s_{0}$ directly.
Nevertheless, we are able to calculate $s_{1}$ and $s_{-1}$ by analyzing the uniform asymptotic behaviors in the neighborhoods of the turning points $\eta_{1}$ and $\eta_{2}$ respectively.
Noting that $s_{-1}=-\overline{s_{1}}$, we only need to calculate $s_{1}$.

\begin{figure}[!ht]
\begin{minipage}{0.48\textwidth}
  \centering\includegraphics[width=5cm, scale=1]{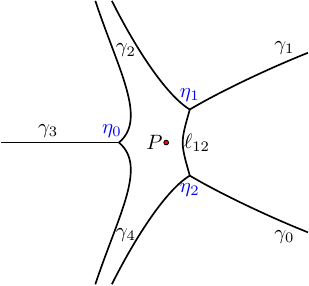}
\end{minipage}
\begin{minipage}{0.48\textwidth}
  \centering\includegraphics[width=5cm, scale=1]{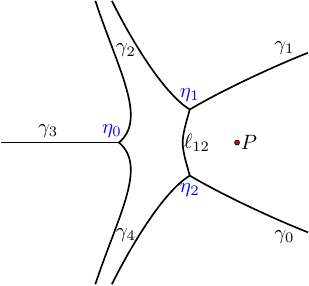}
\end{minipage}
  \caption{Stokes geometry of $Q(\eta,\xi)d\eta^{2}$.}\label{Figure-stokes}
\end{figure}

Following the main ideas in \cite{BCLM}, we can approximate the solutions of (\ref{Schrodinger-equation-scaled-I}) via the Airy functions near $\eta_{1}$ and $\eta_{2}$ respectively. Define two conformal mappings $\zeta(\eta)$ and $\omega(\eta)$ by
\begin{equation}\label{zeta-define}
\int_{0}^{\zeta}s^{\frac{1}{2}}ds=\int_{\eta_{1}}^{\eta}F(s,\xi)^{\frac{1}{2}}ds
\end{equation}
and
\begin{equation}\label{omega-define}
\int_{0}^{\omega}s^{\frac{1}{2}}ds=\int_{\eta_{2}}^{\eta}F(s,\xi)^{\frac{1}{2}}ds,
\end{equation}
respectively from neighborhoods of $\eta=\eta_1$ and $\eta=\eta_2$ to the ones of origin, In the present paper, the principal branches are chosen for all the square roots. Then the conformality can be extended to the Stokes curves, and the following lemma is a consequence of  \cite[Theorem 2]{BCLM}.
\begin{lemma}\label{lem-uniform-case-I}
Let $Y$ be any solution of \eqref{Schrodinger-equation-scaled-I}, then there are constants $C_{1}, C_{2}$ such that
\begin{equation}\label{eq-asym-Y-zeta}
Y=\left(\frac{\zeta}{F(\eta,\xi)}\right)^{\frac{1}{4}}\left\{\left[C_{1}+r_{1}(\eta,\xi)\right]\Ai(\xi^{\frac{2}{3}}\zeta)+\left[C_{2}+r_{2}(\eta,\xi)\right]\Bi(\xi^{\frac{2}{3}}\zeta)\right\},
\end{equation}
where $r_{1}(\eta,\xi),r_{2}(\eta,\xi)=o(|C_{1}|+|C_{2}|)$ as $\xi\rightarrow+\infty$ uniformly for $\eta$ on any two adjacent Stokes lines emanating from $\eta_{1}$ and away from $P$.  Similarly there are constants $\tilde{C}_{1}, \tilde{C}_{2}$ such that
\begin{equation}
Y=\left(\frac{\omega}{F(\eta,\xi)}\right)^{\frac{1}{4}}\left\{\left[\tilde{C}_{1}+\tilde{r}_{1}(\eta,\xi)\right]\Ai(\xi^{\frac{2}{3}}\omega)+\left[\tilde{C}_{2}+\tilde{r}_{2}(\eta,\xi)\right]\Bi(\xi^{\frac{2}{3}}\omega)\right\},
\end{equation}
where $\tilde{r}_{1}(\eta,\xi),\tilde{r}_{2}(\eta,\xi)=o(|\tilde{C}_{1}|+|\tilde{C}_{2}|)$ as $\xi\rightarrow+\infty$ uniformly for $\eta$ on any two adjacent Stokes lines emanating from $\eta_{2}$ and away from $P$.
\end{lemma}

\begin{remark}\label{remark-r-tilde-r}
It can be seen from the proof of \cite[Theorem 2]{BCLM} that the limit values $\lim\limits_{\eta\to\infty}r_{i}(\eta,\xi)$, $i=1,2$ with $\arg\eta\sim\frac{(2k-1)\pi}{5}$, $k=1,2$ all exist. Similar property is applicable to $\tilde{r}_{1}(\eta,\xi)$ and $\tilde{r}_{2}(\eta,\xi)$. For convenience, we denote these limit values by $r_{i,k}(\xi),\tilde{r}_{i,k}(\xi)$, $i=1,2,$ $k=1,2$ respectively.
\end{remark}

Moreover, we get the asymptotic behavior of $\zeta(\eta)$ and $\omega(\eta)$ as $\xi,\eta\rightarrow\infty$.
\begin{lemma}\label{lemma-zeta-omega-eta-infty-relation}
As $\eta\to\infty$, the conformal mappings $\zeta(\eta)$ and $\omega(\eta)$ satisfy
\begin{equation}\label{zeta-eta-infty-relation}
\frac{2}{3}\zeta^{\frac{3}{2}}=\frac{4}{5}\eta^{\frac{5}{2}}+E_{0}-\frac{E_{1}}{\xi}+\epsilon_{1}(\xi)+\mathcal{O}(\eta^{-\frac{1}{2}}),
\qquad \arg\eta\in \left(-\frac{\pi}{5},\pi\right),
\end{equation}
and
\begin{equation}\label{omega-eta-infty-relation}
\frac{2}{3}\omega^{\frac{3}{2}}=\frac{4}{5}\eta^{\frac{5}{2}}+F_{0}-\frac{F_{1}}{\xi}+\epsilon_{2}(\xi)+\mathcal{O}(\eta^{-\frac{1}{2}}),
\qquad \arg\eta\in \left(-\pi,\frac{\pi}{5}\right),
\end{equation}
where $\epsilon_{1}(\xi),\epsilon_{2}(\xi)=o\left(\xi^{-1}\right)$ as $\xi\rightarrow\infty$,
and $E_{0}, F_{0}, E_{1}$ and $F_{1}$ are given in (\ref{eq-constants-lemma-1}).
\end{lemma}

The proof of this lemma is the same as that of \cite[Lemma 5]{Long-Li-Liu-Zhao}
except that the exact values of $E_{1},F_{1}$ in \cite[eqs. (28)-(29)]{Long-Li-Liu-Zhao}
should be replaced with the ones in eq.~\eqref{eq-constants-lemma-1}.

\textbf{Proof of Lemma \ref{lem-case-I}}
Unlike the situation considered in \cite{BCLM},
the Stokes multipliers $s_{k}$'s in this paper depend on $\xi$,
hence we will only take the limit $\eta\to\infty$
when we do the asymptotic matching between the Airy functions and $\Phi_{k}$'s.
According to \cite[eqs. (9.2.12), (9.7.5) and (9.2.10)]{NIST-handbook}, we have
\begin{eqnarray}\label{eq-asym-Ai}
\begin{cases}
\Ai(z)\sim
\frac{1}{2\sqrt{\pi}}z^{-\frac{1}{4}}e^{-\frac{2}{3}z^{\frac{3}{2}}},
\quad & \arg{z}\in(-\pi,\pi),\\
\Ai(z)\sim
\frac{1}{2\sqrt{\pi}}z^{-\frac{1}{4}}e^{-\frac{2}{3}z^{\frac{3}{2}}}
+\frac{i}{2\sqrt{\pi}}z^{-\frac{1}{4}}e^{\frac{2}{3}z^{\frac{3}{2}}},
\quad & \arg{z}\in\left(\frac{\pi}{3},\frac{5\pi}{3}\right),\\
\Ai(z)\sim
\frac{1}{2\sqrt{\pi}}z^{-\frac{1}{4}}e^{-\frac{2}{3}z^{\frac{3}{2}}}
-\frac{i}{2\sqrt{\pi}}z^{-\frac{1}{4}}e^{\frac{2}{3}z^{\frac{3}{2}}},
\quad &\arg{z}\in\left(-\frac{5\pi}{3},-\frac{\pi}{3}\right)
\end{cases}
\end{eqnarray}
and
\begin{eqnarray}\label{eq-asym-Bi}
\begin{cases}
\Bi(z)\sim
\frac{i}{2\sqrt{\pi}}z^{-\frac{1}{4}}e^{-\frac{2}{3}z^{\frac{3}{2}}}
+\frac{1}{\sqrt{\pi}}z^{-\frac{1}{4}}e^{\frac{2}{3}z^{\frac{3}{2}}},
\quad&\arg{z}\in\left(-\frac{\pi}{3},\pi\right),\\
\Bi(z)\sim
\frac{i}{2\sqrt{\pi}}z^{-\frac{1}{4}}e^{-\frac{2}{3}z^{\frac{3}{2}}}
+\frac{1}{2\sqrt{\pi}}z^{-\frac{1}{4}}e^{\frac{2}{3}z^{\frac{3}{2}}},
\quad & \arg{z}\in\left(\frac{\pi}{3},\frac{5\pi}{3}\right),\\
\Bi(z)\sim
-\frac{i}{2\sqrt{\pi}}z^{-\frac{1}{4}}e^{-\frac{2}{3}z^{\frac{3}{2}}}
+\frac{1}{\sqrt{\pi}}z^{-\frac{1}{4}}e^{\frac{2}{3}z^{\frac{3}{2}}},
\quad & \arg{z}\in\left(-\pi,\frac{\pi}{3}\right)
\end{cases}
\end{eqnarray}

When $\eta\to\infty$ with $\arg\eta\sim\frac{\pi}{5}$, it follows from \eqref{zeta-eta-infty-relation} that $\arg\zeta\sim\frac{\pi}{3}$. Substituting \eqref{zeta-eta-infty-relation} into \eqref{eq-asym-Ai} and \eqref{eq-asym-Bi},
and noting that $\lambda=\xi^{\frac{2}{5}}\eta$ from \eqref{eq:transform1}, we get
\begin{eqnarray}\label{eq-Ai-Bi-pi/5}
\begin{cases}
\sqrt{2(\lambda-a)}\left(\frac{\zeta}{F(\eta,\xi)}\right)^{\frac{1}{4}}\Ai(\xi^{\frac{2}{3}}\zeta)\sim c_{1}\frac{-1}{\sqrt{2}}\lambda^{-\frac{1}{4}}e^{-\frac{4}{5}\lambda^{\frac{5}{2}}},\\
\sqrt{2(\lambda-a)}\left(\frac{\zeta}{F(\eta,\xi)}\right)^{\frac{1}{4}}\Bi(\xi^{\frac{2}{3}}\zeta)\sim ic_{1}\frac{-1}{\sqrt{2}}\lambda^{-\frac{1}{4}}e^{-\frac{4}{5}\lambda^{\frac{5}{2}}}+2c_{2}\frac{1}{\sqrt{2}}\lambda^{-\frac{1}{4}}e^{\frac{4}{5}\lambda^{\frac{5}{2}}}
\end{cases}
\end{eqnarray}
as $\lambda\to\infty$ (and $\eta\to\infty$ accordingly), where
\begin{equation}
\begin{split}
c_{1}&=-\frac{1}{\sqrt{2\pi}}\xi^{\frac{2}{15}}e^{-\xi E_{0}+E_{1}+\epsilon_{1}(\xi)},\\
c_{2}&=\frac{1}{\sqrt{2\pi}}\xi^{\frac{2}{15}}e^{\xi E_{0}-E_{1}-\epsilon_{1}(\xi)}.
\end{split}
\end{equation}
From \eqref{eq-canonical-solutions}, a straightforward calculation yields
\begin{equation}\label{eq-asym-Phi21-Phi22}
\left((\Phi_{k})_{21},(\Phi_{k})_{22}\right)\sim \left(\frac{1}{\sqrt{2}}\lambda^{-\frac{1}{4}}e^{\frac{4}{5}\lambda^{\frac{5}{2}}},\,
\frac{-1}{\sqrt{2}}\lambda^{-\frac{1}{4}}e^{-\frac{4}{5}\lambda^{\frac{5}{2}}}\right),
\qquad k\in\mathbb{Z}
\end{equation}
as $\lambda\to\infty$.
Substituting \eqref{eq-Ai-Bi-pi/5} into \eqref{eq-asym-Y-zeta}
and then comparing the resulting equation with \eqref{eq-asym-Phi21-Phi22},
we get
\begin{equation}\label{eq-Y-pi/5}
\begin{split}
&\sqrt{2(\lambda-a)Y}\\
=&[C_{1}+r_{1,1}(\xi)]c_{1}(\Phi_{1})_{22}+[C_{2}+r_{2,1}(\xi)][ic_{1}(\Phi_{1})_{22}+2c_{2}(\Phi_{1})_{21}]\\
=&[C_{2}+r_{2,1}(\xi)]2c_{2}(\Phi_{1})_{21}+\{[C_{1}+r_{1,1}(\xi)]c_{1}+[C_{2}+r_{2,1}(\xi)]ic_{1}\}(\Phi_{1})_{22}
\end{split}
\end{equation}
with $\lambda\in\Omega_{1}$.

When $\eta\to\infty$ with $\arg\eta\sim\frac{3\pi}{5}$,
one has $\arg{\zeta}\sim\pi$. Using the asymptotic behavior of the Airy functions with $\arg{z}\sim\pi$ in \eqref{eq-asym-Ai} and \eqref{eq-asym-Bi}, we have
\begin{eqnarray}\label{eq-Ai-Bi-3pi/5}
\begin{cases}
\sqrt{2(\lambda-a)}\left(\frac{\zeta}{F(\eta,\xi)}\right)^{\frac{1}{4}}\Ai(\xi^{\frac{2}{3}}\zeta)\sim c_{1}\frac{-1}{\sqrt{2}}\lambda^{-\frac{1}{4}}e^{-\frac{4}{5}\lambda^{\frac{5}{2}}}+ic_{2}\frac{1}{\sqrt{2}}\lambda^{-\frac{1}{4}}e^{\frac{4}{5}\lambda^{\frac{5}{2}}},\\
\sqrt{2(\lambda-a)}\left(\frac{\zeta}{F(\eta,\xi)}\right)^{\frac{1}{4}}\Bi(\xi^{\frac{2}{3}}\zeta)\sim ic_{1}\frac{-1}{\sqrt{2}}\lambda^{-\frac{1}{4}}e^{-\frac{4}{5}\lambda^{\frac{5}{2}}}+c_{2}\frac{1}{\sqrt{2}}\lambda^{-\frac{1}{4}}e^{\frac{4}{5}\lambda^{\frac{5}{2}}}
\end{cases}
\end{eqnarray}
as $\lambda\to\infty$ ($\eta\to\infty$).
Substituting \eqref{eq-Ai-Bi-3pi/5} into \eqref{eq-asym-Y-zeta}
and comparing the resulting equation with \eqref{eq-asym-Phi21-Phi22}, we get
\begin{equation}\label{eq-Y-3pi/5-C}
\begin{split}
\sqrt{2(\lambda-a)}Y=&\{[C_{1}+r_{1,2}(\xi)]ic_{2}+[C_{2}+r_{2,2}(\xi)]c_{2}\}(\Phi_{2})_{21}\\
&+\{[C_{1}+r_{1,2}(\xi)]c_{1}+[C_{2}+r_{2,2}(\xi)]ic_{1}\}(\Phi_{2})_{22}
\end{split}
\end{equation}
$\lambda\in\Omega_{2}$.
Combining \eqref{eq-Y-pi/5} with \eqref{eq-Y-3pi/5-C}, and noting that
$$\left((\Phi_{2})_{21},(\Phi_{2})_{22}\right)=\left((\Phi_{2})_{21},(\Phi_{2})_{22}\right)\left(\begin{matrix}1&s_{1}\\0&1\end{matrix}\right),$$
we further obtain
\begin{equation}
r_{1,1}(\xi)-r_{1,2}(\xi)=i(r_{2,2}(\xi)-r_{2,1}(\xi))
\end{equation}
and
\begin{equation}
\begin{split}
s_{1}&=\frac{(C_{2}-iC_{1})c_{2}+(2r_{2,1}(\xi)-ir_{1,2}(\xi)-r_{2,2}(\xi))c_{2}}{(C_{1}+iC_{2})c_{1}+r_{1,2}(\xi)c_{1}+r_{2,2}(\xi)ic_{1}}\\
&=-\frac{ic_{2}}{c_{1}}+\frac{2(r_{2,1}(\xi)-r_{2,2}(\xi))c_{2}}{(C_{1}+iC_{2}+r_{1,2}(\xi)+ir_{2,2}(\xi))c_{1}}.
\end{split}
\end{equation}
Since $r_{i,k}(\xi)=o(|C_{1}|+|C_{2}|)$ as $\xi\to+\infty$ for all $i=1,2$ and $k=1,2$. Hence
\begin{equation}
s_{1}=i e^{2\xi E_{0}-2E_{1}-2\epsilon_{1}(\xi)}\left(1+o(1)\right)=i e^{2\xi E_{0}-2E_{1}+o(1)}
\end{equation}
as $\xi\to+\infty$.

Note that $s_{-1}=\overline{s_{1}}$ and $F_{j}=\overline{E_{j}},$ $j=0,1$.
Hence, by taking the complex conjugate of the leading asymptotic behaviour of $s_{1}$
we obtain the one of $s_{-1}$ in \eqref{asym-Stokes-multipliers-case-I}.

\subsection{Case II: $\frac{B(\xi)^2}{4}-A(\xi)^3=-1$}
In this case, equation \eqref{Schrodinger-equation-scaled} becomes
\begin{equation}\label{Schrodinger-equation-scaled-II}
\begin{split}
\frac{d^{2}Y}{d\eta^{2}}=&\xi^2\left[4(\eta^{3}-1)-\frac{B(\xi)}{\xi(\eta-A(\xi))}+\frac{3}{4}\frac{1}{\xi^2(\eta-A(\xi))^2}\right]Y\\
=&\xi^2\tilde{Q}(\eta,\xi)Y,
\end{split}
\end{equation}
which is a Shr\"{o}dinger equation with three simple turning points,
denoted by $\tilde{\eta}_{j},j=1,2,3$, near $\eta=e^{2\pi i/3},e^{-2\pi i/3}$ and $1$ respectively; see Figure \ref{Figure-stokes-II}. According to \cite{EGS-2008}, the limiting state of the Stokes geometry of the quadratic form $Q(\eta,\xi)d\eta^2$ as $\xi\rightarrow+\infty$ is described in Figure \ref{Figure-stokes-II}, where the Stokes curves are defined by
\[
\re\left(\int_{\tilde{\eta}_{j}}^{\eta}\sqrt{\tilde{Q}(s,\xi)}ds\right)=0, \qquad j=1,2,3.
\]

Note that there is a pole $P(A(\xi),0)$ of $\tilde{Q}$ in the sector between $\gamma_{0}$ and $\gamma_{1}$.
Moreover, we find that when $A(\xi)>1$, the pole $P$ is away from $\eta_{0}$;
see the left subplot of Figure \ref{Figure-stokes-II}.
In this case, one may use the Airy function to approximate the solutions of \eqref{Schrodinger-equation-scaled-II} uniformly in the neighborhood of $\gamma_{0}$ and $\gamma_{1}$. When $A(\xi)\to 1$ as $\xi\to+\infty$, the pole $P$ coalesces to the turning point $\eta_{0}$;
see the right subplot in Figure \ref{Figure-stokes-II}.
In this case, according to \cite{Dunster-1990} or \cite{Long-Li-Liu-Zhao},
the Bessel functions are involved.
Although, in the above two cases, appropriate special functions can be chosen to approximate the solutions of \eqref{Schrodinger-equation-scaled-II} uniformly in the neighborhood of $\gamma_{0}$ and $\gamma_{1}$,
we prefer to analyze for the two cases in a unified way since we intend to calculate the leading asymptotic behavior of the Stokes multiplier $s_{0}$ in a unified form for all $A(\xi)\geq 1$ with $\frac{B(\xi)^2}{4}-A(\xi)^3=-1$. According to the method of uniform asymptotics,
we should approximate the solutions of \eqref{Schrodinger-equation-scaled-II} uniformly in the neighborhood of $\gamma_{0}$ and $\gamma_{1}$,
and the approximation should hold in both cases,
which seems difficult yet.
Fortunately, recalling that $s_{0}=i(1+s_{2}s_{3})$ and $s_{2}=s_{-3}=\overline{s_{3}}$
from \eqref{eq-constraints-stokes-multipliers} and \eqref{eq-sk-s-k-relation},
we only need to deal with $s_{2}$.
Hence, it suffices to derive the uniform asymptotics of \eqref{Schrodinger-equation-scaled-II}
in the neighborhoods of the Stokes curves emanating from $\eta_{1}$.

\begin{figure}[!ht]
\begin{minipage}{0.48\textwidth}
  \centering\includegraphics[width=5cm, scale=1]{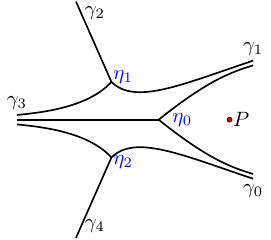}
\end{minipage}
\begin{minipage}{0.48\textwidth}
  \centering\includegraphics[width=5cm, scale=1]{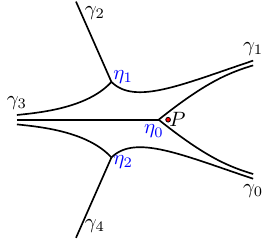}
\end{minipage}
  \caption{Stokes geometry of $\tilde{Q}(\eta,\xi)d\eta^{2}$.}\label{Figure-stokes-II}
\end{figure}

Define the conformal mapping $\theta(\eta)$ by
\begin{equation}\label{zeta-define}
\int_{0}^{\theta}s^{\frac{1}{2}}ds=\int_{\eta_{1}}^{\eta}\tilde{Q}(s,\xi)^{\frac{1}{2}}ds,
\end{equation}
which maps the neighborhood of $\eta=\eta_{1}$ to a neighborhood of the origin.
Then the conformality can be extended to the Stokes curves emanating from $\eta_{1}$
and tending to infinity,
and the following lemma is a consequence of \cite[Theorem 2]{BCLM}.
\begin{lemma}\label{lem-uniform-case-II}
Let $Y$ be any solution of \eqref{Schrodinger-equation-scaled-II}, then there are constants $D_{1}, D_{2}$ such that
\begin{equation}\label{eq-Y-Ai-Bi-theta}
Y=\left(\frac{\theta}{\tilde{Q}(\eta,\xi)}\right)^{\frac{1}{4}}\left\{\left[D_{1}+\bar{r}_{1}(\eta,\xi)\right]\Ai(\xi^{\frac{2}{3}}\theta)+\left[D_{2}+\bar{r}_{1}(\eta,\xi)\right]\Bi(\xi^{\frac{2}{3}}\theta)\right\},
\end{equation}
where $\bar{r}_{1}(\eta, \xi),\bar{r}_{2}(\eta,\xi)=o(|D_{1}|+|D_{2}|)$ as $\xi\to+\infty$ uniformly for $\eta$ on any two adjacent Stokes lines emanating from $\eta_{1}$.
\end{lemma}

\begin{remark}
It is obvious that $\bar{r}_{i}(\eta,\xi),\, i=1,2$,
have properties similar to those of $r_{i}(\eta,\xi)$ mentioned in Remark \ref{remark-r-tilde-r},
and thus we set $\bar{r}_{i,k}(\xi)=\lim\limits_{\eta\to\infty}\bar{r}_{i}(\eta,\xi)$ with $\arg\eta\sim\frac{(2k-1)\pi}{5}$, $k=2,3$.
\end{remark}

Similar to Lemma \ref{lemma-zeta-omega-eta-infty-relation}, we can also obtain the asymptotic behavior of $\theta(\eta)$ as $|\eta|\to\infty$.
\begin{lemma}\label{lemma-theta-eta-infty-relation}
The conformal mapping $\theta(\eta)$ satisfies
\begin{equation}\label{theta-eta-infty-relation}
\frac{2}{3}\theta^{\frac{3}{2}}=\frac{4}{5}\eta^{\frac{5}{2}}+G_{0}-\frac{G_{1}}{\xi}+\epsilon_{3}(\xi)+\mathcal{O}(\eta^{-\frac{1}{2}}),
\qquad \arg\eta\in \left(-\frac{\pi}{5},\pi\right),
\end{equation}
as $|\eta|\rightarrow\infty$, where $\epsilon_{3}(\xi)=o(\xi^{-1})$ as $\xi\to+\infty$ and $G_{0}=\frac{2e^{\frac{2}{3}\pi i}}{5}\B\left(\frac{1}{2},\frac{1}{3}\right)$.
\end{lemma}

\begin{proof}
The argument for deriving the asymptotics of the Stokes multiplier $s_{2}$
is similar to that for $s_{1}$ in Case I $\frac{B(\xi)^2}{4}-A(\xi)^3=1$.
The only difference is that we should apply the uniform asymptotics
of the Airy functions $\Ai(z)$ and $\Bi(z)$ as $z\to\infty$
with $\arg{z}\sim\pi$ and $\arg{z}\sim\frac{5\pi}{3}$.

When $\eta\to\infty$ with $\arg\eta\sim\frac{3\pi}{5}$, one has $\arg{\theta}\sim\pi$.
Substituting \eqref{theta-eta-infty-relation} into the asymptotics of the Airy functions with $\arg{z}\sim\pi$ in \eqref{eq-asym-Ai} and \eqref{eq-asym-Bi},
and recalling the transformation $\lambda=\xi^{\frac{2}{3}}\eta$ from \eqref{eq:transform1}, we have
\begin{eqnarray}\label{eq-Ai-Bi-3pi/5-theta}
\begin{cases}
\sqrt{2(\lambda-a)}\left(\frac{\theta}{F(\eta,\xi)}\right)^{\frac{1}{4}}\Ai(\xi^{\frac{2}{3}}\theta)\sim d_{1}\frac{-1}{\sqrt{2}}\lambda^{-\frac{1}{4}}e^{-\frac{4}{5}\lambda^{\frac{5}{2}}}+id_{2}\frac{1}{\sqrt{2}}\lambda^{-\frac{1}{4}}e^{\frac{4}{5}\lambda^{\frac{5}{2}}},\\
\sqrt{2(\lambda-a)}\left(\frac{\theta}{F(\eta,\xi)}\right)^{\frac{1}{4}}\Bi(\xi^{\frac{2}{3}}\theta)\sim id_{1}\frac{-1}{\sqrt{2}}\lambda^{-\frac{1}{4}}e^{-\frac{4}{5}\lambda^{\frac{5}{2}}}+d_{2}\frac{1}{\sqrt{2}}\lambda^{-\frac{1}{4}}e^{\frac{4}{5}\lambda^{\frac{5}{2}}}
\end{cases}
\end{eqnarray}
as $\lambda\to\infty$ (and $\eta\to\infty$ accordingly) with $\lambda\sim\frac{3\pi}{5}$, where
\begin{equation}
\begin{split}
d_{1}&=-\frac{1}{\sqrt{2\pi}}\xi^{\frac{2}{15}}e^{-\xi G_{0}+G_{1}+\epsilon_{3}(\xi)},\\
d_{2}&=\frac{1}{\sqrt{2\pi}}\xi^{\frac{2}{15}}e^{\xi G_{0}-G_{1}-\epsilon_{3}(\xi)}.
\end{split}
\end{equation}
Combining \eqref{eq-asym-Phi21-Phi22}, \eqref{eq-Y-Ai-Bi-theta} and \eqref{eq-Ai-Bi-3pi/5-theta}, we get
\begin{equation}\label{eq-Y-3pi/5-D}
\begin{split}
\sqrt{2(\lambda-a)}Y=&\{[D_{1}+\bar{r}_{1,2}(\xi)]id_{2}+[D_{2}+\bar{r}_{2,2}(\xi)]d_{2}\}(\Phi_{2})_{21}\\
&+\{[D_{1}+\bar{r}_{1,2}(\xi)]d_{1}+[D_{2}+\bar{r}_{2,2}(\xi)]id_{1}\}(\Phi_{2})_{22}
\end{split}
\end{equation}
with $\lambda\in\Omega_{2}$.

When $\eta\to\infty$ with $\arg{\eta}\sim\pi$, we have $\arg{\theta}\sim\frac{3\pi}{5}$.
Noting that $\Ai(z)=\Ai(ze^{-2\pi i})$ and $\Bi(z)=\Bi(ze^{-2\pi i})$,
it follows from \eqref{eq-asym-Ai} and \eqref{eq-asym-Bi} that
\begin{equation}\label{eq-asym-Ai-Bi-5pi/3}
\Ai(z)\sim\frac{i}{2\sqrt{\pi}}z^{-\frac{1}{4}}e^{\frac{2}{3}z^{\frac{3}{2}}},\quad \Bi(z)\sim\frac{i}{\sqrt{\pi}}z^{-\frac{1}{4}}e^{-\frac{2}{3}z^{\frac{3}{2}}}+\frac{1}{2\sqrt{\pi}}z^{-\frac{1}{4}}e^{\frac{2}{3}z^{\frac{3}{2}}}
\end{equation}
as $z\to\infty$ with $\arg{z}\sim\frac{5\pi}{3}$.
Substituting \eqref{theta-eta-infty-relation} into \eqref{eq-asym-Ai-Bi-5pi/3}, we have
\begin{eqnarray}\label{eq-Ai-Bi-pi-theta}
\begin{cases}
\sqrt{2(\lambda-a)}\left(\frac{\theta}{F(\eta,\xi)}\right)^{\frac{1}{4}}\Ai(\xi^{\frac{2}{3}}\theta)\sim id_{2}\frac{-1}{\sqrt{2}}\lambda^{-\frac{1}{4}}e^{\frac{4}{5}\lambda^{\frac{5}{2}}},\\
\sqrt{2(\lambda-a)}\left(\frac{\theta}{F(\eta,\xi)}\right)^{\frac{1}{4}}\Bi(\xi^{\frac{2}{3}}\theta)\sim 2id_{1}\frac{-1}{\sqrt{2}}\lambda^{-\frac{1}{4}}e^{-\frac{4}{5}\lambda^{\frac{5}{2}}}+d_{2}\frac{1}{\sqrt{2}}\lambda^{-\frac{1}{4}}e^{\frac{4}{5}\lambda^{\frac{5}{2}}}
\end{cases}
\end{eqnarray}
as $\lambda\to\infty$ ($\eta\to\infty$) with $\lambda\sim\pi$.
A combination of \eqref{eq-asym-Phi21-Phi22}, \eqref{eq-Y-Ai-Bi-theta} and \eqref{eq-Ai-Bi-pi-theta},
yields
\begin{equation}
i(\bar{r}_{1,2}-\bar{r}_{1,3})+(\bar{r}_{2,2}-\bar{r}_{2,3})=0
\end{equation}
and
\begin{equation}\label{eq-Y-pi-D}
\begin{split}
\sqrt{2(\lambda-a)}Y=&\{[D_{1}+\bar{r}_{1,3}(\xi)]id_{2}+[D_{2}+\bar{r}_{2,3}(\xi)]d_{2}\}(\Phi_{3})_{21}\\
&+[D_{2}+\bar{r}_{2,3}(\xi)]2id_{1}(\Phi_{3})_{22}
\end{split}
\end{equation}
with $\lambda\in\Omega_{3}$. Combining \eqref{eq-Y-3pi/5-D} with \eqref{eq-Y-pi-D} and noting that
$$\left((\Phi_{3})_{21},(\Phi_{3})_{22}\right)=\left((\Phi_{2})_{21},(\Phi_{2})_{22}\right)\left(\begin{matrix}1&0\\s_{2}&1\end{matrix}\right),$$
we have
\begin{equation}
\begin{split}
s_{2}
=&\frac{[D_{1}+\bar{r}_{1,2}(\xi)]d_{1}+[D_{2}+\bar{r}_{2,2}(\xi)]id_{1}-[D_{2}+\bar{r}_{2,3}(\xi)]2id_{1}}{[D_{1}+\bar{r}_{1,3}(\xi)]id_{2}+[D_{2}+\bar{r}_{2,3}(\xi)]d_{2}}\\
=&\frac{d_{1}}{id_{2}}+\frac{d_{1}}{id_{2}}\frac{2i(\bar{r}_{2,2}(\xi)-\bar{r}_{2,3}(\xi))}{D_{1}-iD_{2}+\bar{r}_{1,3}(\xi)-i\bar{r}_{2,3}(\xi)}.
\end{split}
\end{equation}
In view of the fact that $\bar{r}_{i,k}(\xi)=o(|D_{1}|+|D_{2}|),i=1,2,k=2,3$ as $\xi\to+\infty$,
we immediately get the leading asymptotic behavior of $s_{2}$ in \eqref{asym-Stokes-multipliers-case-II}.
Note that $s_{3}=-\overline{s_{2}}$,
the leading asymptotic behavior of $s_{3}$ is obtained by taking the complex conjugate
of the one of $s_{2}$.
\end{proof}

\section{Discussion}
In the present paper, we study the connection problem of the PI equation between the initial data and the large negative $t$ asymptotic behaviors.
We find how the solutions evolve when we let both $y(0)=a$ and $y'(0)=b$ vary.
Precisely speaking,
we show that there exists a sequence of curves $\Gamma_{n}$, $n=1,2,\dots$, on the $(a,b)$-plane,
corresponding to the real tronqu\'{e}e solutions, \textit{i.e.}, separatrix solutions.

When $(a,b)$ lies between two curves $\Gamma_{2m}$ and $\Gamma_{2m+1}$, $m=1,2,\cdots$,
the solution $y(t;a,b)$ oscillates around $y(t)=-\sqrt{-t/6}$;
and when $(a,b)$ lies between $\Gamma_{2m-1}$ and $\Gamma_{2m}$, $m=1,2,\cdots$,
the solution $y(t;a,b)$ has infinite number of poles on the negative real axis.

One may compare our results with the ones in a previous work \cite{Long-Li-Liu-Zhao}.
Note that, in \cite{Long-Li-Liu-Zhao},
the authors assume that one of the initial data $a$ and $b$ is fixed and the other is large,
hence \cite{Long-Li-Liu-Zhao} studies how the PI solution $y(t;a,b)$ evolves when the initial data $(a,b)$ vary along a horizontal or a vertical line on the $(a,b)$-plane.
While in the present paper, we assume that $|\frac{b^2}{4}-a^3|$ is large and study how the solution evolves when the initial data $(a,b)$ vary along any directions.
Therefore, the results in the present paper include the ones in \cite{Long-Li-Liu-Zhao}
as special cases, as well as provide new results.

A similar study to the second Painlev\'{e} equation is carried out in \cite{Long-Zeng}. They give an asymptotic classification of the PII solutions
with large initial data $(a,b)$ by assuming that $b^2-a^4\to\pm\infty$.
However, they only studied the classification problem of the PII solutions when $a^4|b^2-a^4|^{-1}$ is bounded,
which is similar to the first case $a^{3}|\frac{b^2}{4}-a^3|^{-1}\leq M_{1}$ in the present paper.
We believe there is also another interesting case for PII to be analyzed
and we expect a result of PII similar to that of PI.
After Bender--Komijani's work on the first and second Painlev\'e transcendents~\cite{Bender-Komijani-2015},
they obtained a similar result to the fourth Painlev\'e equation PVI~\cite{Bender-Komijani-2022}.
It is also interesting to apply the same analysis in the present paper
to the initial value problems in the more general setting for PVI,
namely, one may consider the case when both $a$ and $b$ are large instead of one of them being large.

An interesting problem was brought out by Bender et al.~\cite{Bender-Komijani-Wang-2019}
for the generalized Painlev\'e equations, where they found the nonlinear eigenvalues
giving rises to separatrix solutions.
Since our method starts with the Lax pair or the Schrodinger equation,
which is not known to the generalized Painlev\'e equations yet;
thus this method can not apply to generalized Painlev\'e equations at the moment.

The classification of the PI solutions obtained in the present paper
is an asymptotic one, namely, we assume the initial data $(a,b)$ is large,
and we only obtain the limiting form equation of the curves $\Gamma_{n}$ as $n\to+\infty$.
A full solution to the problem of finding the exact equation of $\Gamma_{n}$ for finite $n$ is still an open problem and deserves further studies.

\section*{Acknowledgement}
All authors are grateful to the Reviewers and the Editor for their invaluable comments and suggestions, which improve the readability of the current paper significantly.
This work was supported by
the National Natural Science Foundation of China [Grant nos. 12071394],
the Natural Science Foundation of Hunan Province [Grant no. 2020JJ5152],
the General Project of Hunan Provincial Department of Education [Grant no. 19C0771],
and the Doctoral Startup Fund of Hunan University of Science and Technology [Grant no. E51871].

\end{document}